\documentclass[11pt]{amsart}
\usepackage{amssymb, latexsym}

\newcommand{\beqa}{\begin{eqnarray*}}
\newcommand{\eeqa}{\end{eqnarray*}}
\newcommand{\beqn}{\begin{eqnarray}}
\newcommand{\eeqn}{\end{eqnarray}}

\newcommand{\iy}{\infty}
\newcommand{\lt}{\left}
\newcommand{\rt}{\right}

\newcommand{\C}{\mathbb C}
\newcommand{\R}{\mathbb R}

\newcommand{\N}{\mathbb N}

\newcommand{\mcH}{\mathcal H}
\newcommand{\mcB}{\mathcal B}
\newcommand{\mcC}{\mathcal C}

\newcommand{\tf}{\tfrac}

\newcommand{\e}{\varepsilon}

\newcommand{\De}{\Delta}
\newcommand{\la}{\lambda}

\newcommand{\Om}{\Omega}

\newcommand{\s}{\sigma}
\newcounter{cnt1}
\newcounter{cnt2}
\newcounter{cnt3}
\newcommand{\blr}{\begin{list}{$($\roman{cnt1}$)$}
 {\usecounter{cnt1} \setlength{\topsep}{0pt}
 \setlength{\itemsep}{0pt}}}
\newcommand{\bla}{\begin{list}{$($\alph{cnt2}$)$}
 {\usecounter{cnt2} \setlength{\topsep}{0pt}
 \setlength{\itemsep}{0pt}}}
\newcommand{\bln}{\begin{list}{$($\arabic{cnt3}$)$}
 {\usecounter{cnt3} \setlength{\topsep}{0pt}
 \setlength{\itemsep}{0pt}}}
\newcommand{\el}{\end{list}}
\newtheorem{thm}{Theorem}[section]

\newtheorem{cor}[thm]{Corollary}
\newtheorem{ex}[thm]{Example}

\newtheorem{Def}[thm]{Definition}

\newtheorem{rem}[thm]{Remark}
\newcommand{\Rem}{\begin{rem} \rm}
\newcommand{\bdfn}{\begin{Def} \rm}
\newcommand{\edfn}{\end{Def}}

\newcommand{\ba}{\begin{array}}
\newcommand{\ea}{\end{array}}

\newtheorem*{acknowledgements}{Acknowledgements}
\numberwithin{equation}{section}
\date{}
\begin{document}
\title{\bf{Note on the  Spectral Theorem }}
\author[Gill]{T. L. Gill}
\address[Tepper L. Gill]{ Departments of Mathematics, Physics, and Electrical \& Computer Engineering, Howard University\\
Washington DC 20059 \\ USA, {\it E-mail~:} {\tt tgill@howard.edu}}
\author[Williams]{D.  Williams}
\address[Daniel Williams]{ Department of Mathematics \\ Howard
University\\ Washington DC 20059 \\ USA, {\it E-mail~:} {\tt
dwilliams@howard.edu}}
\date{}
\maketitle
\begin{abstract}    In this note, we show that the spectral theorem,  has two representations;  the Stone-von Neumann representation and one based on the polar decomposition of linear operators, which we call the deformed representation.  The deformed representation has the advantage that it provides an easy extension to all closed densely defined linear operators on Hilbert space.  Furthermore, the deformed representation can also be extended all separable reflexive  Banach spaces and has a limited extension to non-reflexive Banach spaces. 
\end{abstract}
\section*{Introduction}
Let $\mcC[\mcB]$ be the closed densely defined linear operators on a Banach space.  By definition an operator $A$, defined on a separable Banach space $\mcB$ is of Baire class one if it can be approximated by a sequence $\{A_n \} \subset L[\mcB]$, of bounded linear operators. If $\mcB$ is a Hilbert space, then every $A \in \mcC[\mcB]$ is of Baire class one.  However, it turns out  that, if $\mcB$ is not a Hilbert space, there may be operators $A \in \mcC[\mcB]$ that are not of Baire class one.  

A Banach space $\mcB$ is said to be quasi-reflexive if $ \dim \{ {\mcB}'' /{\mcB}\}  < \infty $, and  nonquasi-reflexive if $ \dim \{ {\mcB}'' /{\mcB}\} = \infty $.  Vinokurov, Petunin and Pliczko \cite{VPP} have shown that, for every nonquasi-reflexive Banach space $\mcB$, there is a closed densely defined linear operator $A$ which is not of Baire class one (for example, $C[0,1]$ or $L^1[{\R}^n], \; n \in \N$).  It can even be arranged so that $A^{-1}$ is a bounded linear injective operator (with dense range).  This means that there does not exist a sequence of bounded linear operators $A_{n} \in L[\mcB]$, with $A_n \phi \rightarrow A\phi$ for each $A \in \mcC[\mcB]$ and $\phi \in D(A)$.  

In \cite{GSZ}, we were able to show that if $\mcB$ is one of the classical Banach spaces,  then each bounded linear operator  has an adjoint and a spectral type representation, which we called the extended representation.   In addition, we were able to show that an operator has a adjoint if and only if it is of Baire class one.  At that time, it was not clear if our results could be extended to the closed densely defined linear operators on any of the classical reflexive Banach spaces.  However, we were able to show that every generator of a $C_0$-semigroup on any classical space has an adjoint (is of Baire class one).  We were later informed by Professor Pliczko that the results of their paper automatically implies that the set of closed densely defined linear operators on every separable reflexive (or quasi-reflexive) Banach is of Baire class one.
\section*{Purpose}
In the first section of this paper we show how the polar decomposition property of operators leads to a new type of (deformed) spectral representation for linear operators, which easily extends to all linear operators in ${\mcC}[{\mcH}]$.   We then show directly that, for any separable reflexive Banach space,  every closed densely defined linear operator has an adjoint. We use this result to prove  that the closed densely defined linear operators are of Baire class one.   Finally, we show that an almost identical (deformed) spectral representation holds if we replace $\mcH$ by a reflexive Banach space.
\section*{Preliminaries} 
If $\mcH$ is a Hilbert space and $A$ is any selfadjoint operator in ${\mcC}[{\mcH}]$, the following direct spectral representation theorem is well-known (see \cite{DS}, page 1192-99). 
\begin{thm} Let $A \in {\mcC}[{\mcH}]$ be a selfadjoint operator. Then its spectrum $\s(A) \subset \R$ and there exists a unique regular countably additive  projection-valued spectral  measure ${\bf{E}}(\cdot)$ mapping the Borel sets, ${\mathfrak{B}}[\R]$, over $\R$ into $\mcH$ such that
\begin{enumerate}
\item $D(A)$ satisfies
\[
D(A) = \left\{ {\left. {\phi \in \mathcal{H}} \ \right|\;\int_{-\iy}^\iy {\lambda ^2 \left( {{\bf{E}}(d\lambda )\phi,\phi} \right)_\mathcal{H}  < \infty } } \right\}
\]
and 
\[
A\phi =  \int_{ - \iy}^\iy {\lambda {\bf{E}}(d\lambda )\phi}, \; {\rm for \; each } \; \phi \in D(A). 
\]
\item If $g(\cdot)$ is a complex-valued Borel function defined (a.e) on $\R$, then $g(A)\in {\mcC}[{\mcH}]$ and, for  $ \phi \in D(g(A))= D_g(A)$,  
\[
g(A)\phi=  \int_{ - \iy}^\iy {g(\lambda) {\bf{E}}(d\lambda )\phi},
\]
 where
\[
D_g(A) = \left\{ {\left. {\phi \in \mathcal{H}} \ \right|\;\int_{-\iy}^\iy {\lt|g(\lambda)\rt|^2 \left( {{\bf{E}}(d\lambda )\phi,\phi} \right)_\mathcal{H}  < \infty } } \right\}
\]
and $g(A^*)= {\bar g}(A)$.
\end{enumerate}
\end{thm} 
\begin{rem} 
We call Theorem 0.1 the direct representation because it requires the least additional material compared to the Gelfand representation and the one favored in mathematical physics, based on the position representation (see Rudin \cite {RU} page 306 and Reed and Simon \cite{RS}, page 260).
\end{rem} 
 A standard theorem of von Neumann [VN] shows that on ${\mathcal{H}}$, any closed densely defined linear operator $A$ has a well-defined adjoint $A^{*}$ with $A^{*}A$ nonnegative and selfadjoint.  A  classic result shows that there is a unique partial isometry $U$  such that $A=UT={\bar T}U$ and  $A^*=U^*{\bar T} =TU^* $, where $T=[A^{*}A]^{1/2}, \ {\bar T}=[AA^*]^{1/2}$, and $D(A)=D(T)$  (see Kato \cite{K}, page 334).  The next result can be found in Hille and Phillips \cite{HP} (see page 63). 
\begin{thm} Let ${\bf{G}}(\la)$ be a vector-valued function from $\R$ to $\mcH$ of bounded variation.  If $h(\la)$ is a continuous complex-valued function on $(a,b) \subset \R$,  then the following holds:
\begin{enumerate}
\item The integral $\int_a^b {h(\lambda )} d{\bf{G}}(\lambda )$ exists in the  $\mcH$ norm.

\item  If $T \in \mcC[{\mcH}], \; {\bf G}(\la) \in D(T)$ and $T{\bf{G}}(\lambda )$ is of bounded variation then
\[
T\int_a^b {h(\lambda )} d{\bf G}(\lambda )=\int_a^b {h(\lambda )} dT{\bf G}(\lambda).
\]
\end{enumerate}
\end{thm}   
\section{Hilbert Space Theory}
\begin{Def}If $U$ is a partial isometry and ${\bf E}(\cdot)$ is a positive spectral measure, then we call ${\bf{F}}(\cdot)=U{\bf{E}}(\cdot)$ a deformed spectral measure.
\end{Def}
\begin{thm} Let $A \in {\mcC}[{\mcH}]$ be arbitrary.  Then, for each $\phi \in D(A)$, there exists a deformed spectral measure ${\bf{F}}(\cdot)$ such that:
\begin{enumerate}
\item $D(A)$  satisfies
\[
D(A) = \left\{ {\left. {\phi \in \mathcal{H}} \ \right|\;\int_{0}^\iy {\lambda ^2 \left({d{\bf{F}}(\lambda )\phi,\phi} \right)_\mathcal{H}  < \infty } } \right\}
\]
and
\item 
\[
A\phi =  \int_{ 0}^\iy {\lambda d{\bf{F}}(\lambda )\phi}, \ {\rm for \ all} \; \phi \in D(A). 
\]
\item If $g(\cdot)$ is a complex-valued Borel function defined (a.e) on $\R$, then 
\[
g(A)\phi= \int_{ 0}^\iy {g(\lambda) d{\bf{F}}(\lambda )\phi}  \ {\rm for \ all} \; \phi \in D_g(A),
\]
where
\[
D_g(A) = \left\{ {\left. {\phi \in \mathcal{H}} \ \right|\;\int_{0}^\iy {\lt|g(\lambda)\rt|^2 \left( {d{\bf{F}}(\lambda )\phi,\phi} \right)_\mathcal{H}  < \infty } } \right\}.
\]
\end{enumerate}
\end{thm}
\begin{proof}  To prove (1), write $A=UT$, where $U$ is the unique partial isometry and $T=[A^{*}A]^{1/2}$.  By Theorem 0.1, there is a positive spectral measure ${\bf{E}}(\cdot)$ such that, for each $x \in D(A)=D(T)$:   
\beqn
T\phi=  \int_{ 0}^\iy {\lambda} d{\bf{E}}(\lambda )\phi.
\eeqn
Since ${\bf{E}}(\lambda )\phi$ is a positive vector-valued function of bounded variation and  $U$ is a partial isometry,    ${\mathbf{F}}(\lambda )x=U{\mathbf{E}}(\lambda )\phi$ is of bounded variation, with $Var({\mathbf{F}}\phi,\R) \le Var({\mathbf{E}}\phi,\R)$.  Thus, by Theorem 0.2, 
\[
U\int_0^\iy {\lambda} d{\mathbf{E}}(\lambda )\phi=\int_0^\iy {\lambda} dU{\mathbf{E}}(\lambda )\phi.
\]
Since $A\phi=UT\phi$, if we set ${\mathbf{F}}(\lambda )x=U{\mathbf{E}}(\lambda )x$, we have from equation (1.1),   
\beqn
A\phi=  \int_{0}^\iy {\lambda} d{\mathbf{F}}(\lambda )\phi.
\eeqn
The proof of (2) and (3) follows the proof of the same results in \cite{DS}, when we recall that $\lt\|A\phi\rt\|^2=\lt\|T\phi\rt\|^2$, so that
\[
\begin{gathered}
  \left\{ {\phi\left| {\;\int_{ 0 }^\iy {\lambda ^2 \left( {d{\mathbf{E}}(\lambda )\phi ,\phi } \right) < } \infty } \right.} \right\} = \left\{ {\phi\left| {\;\int_{0 }^\infty {\lambda ^2 \left( {d{\mathbf{F}}(\lambda )\phi ,\phi } \right) < } \infty } \right.} \right\} \hfill \\
  \left\{ {\phi\left| {\;\int_{ 0 }^\infty {\left| {g(\lambda )} \right|^2 \left( {d{\mathbf{E}}(\lambda )\phi ,\phi } \right) < } \infty } \right.} \right\} = \left\{ {\phi\left| {\;\int_{ 0 }^\infty {\left| {g(\lambda )} \right|^2 \left( {d{\mathbf{F}}(\lambda )\phi ,\phi } \right) < } \infty } \right.} \right\}. \hfill \\ 
\end{gathered} 
\] 
\end{proof} 
\begin{rem} Let $A$ be self adjoint, with its spectrum on the negative real axis.  In this case, the standard spectral theorem gives us:
\beqn
A = \int_{-r_A}^0 {\lambda d{\bf E}(\lambda )} .
\eeqn
However, the deformed spectral theorem gives 
\beqn
A = \int_{0}^{r_A} {\lambda d{\bf F}(\lambda )} .
\eeqn
In fact, the domain of integration for ${\bf F}(\la)$ always coincides with the spectrum of the positive linear operator $T$, where $A=UT$. Thus, if $A$ is not a positive selfadjoint linear operator, the two representations are distinct.
\end{rem}
\section{Adjoints on Reflexive Banach Spaces}
In this section, we define the adjoint for operators in $\mcC[\mcB]$, but first,  we collect a few tools, which are needed for the sequel.  
Recall that duality map $J: \mcB \mapsto {\mcB}'$, is the set  
\[
J(u) = \left\{ {{F_u} \in \mcB'\left| { F_u(u)=\left\langle {u,{F_u}} \right\rangle  = {{\left\| u \right\|}^2} = {{\left\| {{F_u}} \right\|}^2}} \right.} \right\},\;\forall u \in \mcB.
\]
We want to construct a special duality map. Assume that  $\mcB \subset \mcH$ is a dense continuous embedding and fix $u \in \mcB$.   Let $M= \left\langle {u} \right\rangle$ be the closed subspace spanned by $u$ and define a seminorm $p_u ( \ \cdot \ )$ on $\mcB$ by  $p_u (v) = \lt\| u \rt\|_{\mcB} \lt\| v \rt\|_{\mcB}$.  Define the map $ {\hat S}_u (\, \cdot \,) =\left\langle {\, \cdot \,,{\hat S}_u } \right\rangle$ by:
\[
\left\langle { v ,{\hat S}_u } \right\rangle = {\hat S}_u (v) = \frac{{\left\| u \right\|_{\mcB}^2 }}
{{\left\| u \right\|_2^2 }}\left( {v,u} \right)_2 .
\]
On the closed subspace $M= \left\langle {u} \right\rangle, \; \left| {{\hat S}_u (v)} \right| = \left\| u \right\|_B \left\| v \right\|_B  \leqslant p_u (v)$.  By the Hahn-Banach Theorem, ${\hat S}_u (\, \cdot \,)$ has an extension, $S_u (\, \cdot \,)$, to $\mcB$ such that $\left| {S_u (v)} \right|   \leqslant p_u (v)= \left\| u \right\|_B \left\| v \right\|_B$ for all $v \in \mcB$. From here, we see that $\left\| {S_u } \right\|_{\mathcal{B}'} \le \left\| u \right\|_\mathcal{B}$.   On the other hand, we have $\left\| u \right\|_\mathcal{B}^2=S_u (u)  \leqslant \left\| u \right\|_\mathcal{B} \left\| {S_u } \right\|_{\mathcal{B}'}$, 
so that $S_u ( \, \cdot \, )$ is a duality mapping for $u$.  We call $S_u ( \, \cdot \, )$ the {\it Steadman duality map} on $\mathcal{B}$ associated with ${\mathcal{H}}$. 

Recall that on $\mcB$, a densely defined operator  ${A}$ is called accretive if $\operatorname{Re} \left\langle {Au ,F_u} \right\rangle  \ge 0$ for $u \in D(A)$ and any duallity map $F_u$.  The following definition extracts the essential properties of an adjoint operator on Hilbert space, but also makes sense on a Banach space
\begin{Def} If $A \in \mcC[\mcB]$, the closed densely defined linear operators on $\mcB$, we say that $A^*$ is a adjoint of $A$ if:
\begin{enumerate}
\item the operator $
A^ * A \ge 0$  (accretive),		
\item $(A^ * A)^ *   = A^ * A$ (naturally selfadjoint), and 
\item $I + A^ * A$ has a bounded inverse.
\end{enumerate}
\end{Def}
We need the following result by Lax \cite{L}.  
\begin{thm}[Lax]\label{L: lax}  Suppose ${\mathcal{B}}$ is a dense continuous embedding in a separable Hilbert space ${\mathcal{H}}$.   Let $A \in L[{\mathcal{B}}]$.  If $A$ is selfadjoint on ${\mathcal{H}}$ (i.e., $\left( {Au,v} \right)_{\mathcal{H}}  = \left( {u,Av} \right)_{\mathcal{H}}, 
\forall u{\text{,}}v \in {\mathcal{B}}$), then  $A$ is bounded on ${\mathcal{H}}$ and $
\left\| A \right\|_{\mathcal{H}}  \leqslant k \left\| A \right\|_{\mathcal{B}}$ for some positive constant $k$.
\end{thm} 
The following general result is a variation of one due to Kuelbs \cite{KB}.
\begin{thm} \label{GK} Suppose ${\mathcal{B}}$ is a separable reflexive Banach space, then there exist a separable Hilbert space ${\mathcal{H}}$  such that, 
\begin{enumerate}
\item
${\mathcal{B}} \subset {\mathcal{H}}$ as a continuous dense embedding.
\item 
${\mathcal{B}'} \subset {\mathcal{H}'}$ as a continuous dense embedding.
\end{enumerate}
\end{thm}
\begin{proof} To prove (1), let $\{ u_n \} $ be a dense set in ${\mathcal{B}}$ and let $\{ f_n \} $ be any fixed set of corresponding duality mappings (i.e., $f_n  \in {\mathcal{B'}}$, the dual space of ${\mathcal{B}}$ and $
f_n (u_n ) = \left\langle {u_n ,f_n } \right\rangle  = \left\| {u_n } \right\|_{\mathcal{B}}^2 = \left\| {f_n } \right\|_{\mathcal{B}'}^2 $).   Let $\{ t_n \}$ be a positive sequence of numbers such that $\sum\nolimits_{n = 1}^\infty  {t_n }  = 1$, and define $\left( {u,v} \right)$ by:
\[
\left( {u,v} \right) = \sum\nolimits_{n = 1}^\infty  {t_n f_n (u)} \bar f_n (v).
\]
It is easy to see that $\left( {u,v} \right)$ is an inner product on ${\mathcal{B}}$.  Let $
{\mathcal{H}}$ be the completion of ${\mathcal{B}}$ with respect to this inner product.  It is clear that ${\mathcal{B}}$ is dense in ${\mathcal{H}}$, and 
\[
\left\| u \right\|^2  = \sum\nolimits_{n = 1}^\infty  {t_n \left| {f_n (u)} \right|^2 }  \le \sup _n \left| {f_n (u)} \right|^2  = \left\| u \right\|_{\mathcal{B}}^2,
\]
so the embedding is continuous.

For (2), we note that,  since $\mcB$ is reflexive $\mcB=\mcB''$.  In this case the set $\{ f_n \} $ is dense in ${\mathcal{B}'}$, so we may use the dense family $\{ u_n \} \subset {\mathcal{B}}$ to define an inner product on $\mcB'$ by 
\[
\left( {f,g} \right)  = \sum\nolimits_{n = 1}^\infty  {t_n u''_n (f)} \bar u''_n (g) = \sum\nolimits_{n = 1}^\infty  {t_n f (u_n)} \bar{g}(u_n),
\]
where $u''_n (h)=h(u_n)$, for each $h \in \mcB'$.  The completion of $\mcB'$ with the above inner product  provides a construction of $\mcH'$.  It is clear that $\mcB' \subset \mcH'$ as a continuous dense embedding.
\end{proof}
\begin{thm}When $\mcB$ is reflexive every operator $A \in \mcC[\mcB]$ has a well defined adjoint $A^* \in \mcC[\mcB]$.  If $A \in L[\mcB]$, the bounded linear operators on $\mcB$, then $A^* \in L[\mcB]$.
\end{thm}
 \begin{proof}
Let ${\bf J} :{\mcH}  \to {\mcH}'$.  It is easy to see that  ${\bf J}^*={\bf J}$.  Since $[{\bf J}]_{\mcB}$ is one to one and onto $\mcB'$, if ${A}  \in \mcC[{\mcB}]$, then $[ {{A'} {\bf J} } ]_{\mcB}: \mcB' \to \mcB'$. Since $A'$ is closed and densely defined, it follows that ${\bf J}^{ - 1} {A'} {\bf J} :{\mcB}  \to {\mcB}$ is closed and densely defined. Thus, we can define ${A}^ {*}   = [ {{\bf J}^{ - 1} {A'} {\bf J} } ]_{\mcB}$.

In case ${A}  \in L[{\mcB}]$, we know that ${A}^ {*}   = [ {{\bf J}^{ - 1} {A'} {\bf J} } ]_{\mcB}$ is defined on all of $\mcB$.  By the closed graph theorem,   ${A}^ {*} \in L[\mcB]$.  

In both cases, $A^*$ and $[A^*]'$ are also densely defined on ${\mcH}$ and  ${\mcH}'$ respectively, so that we can extend $A$ and $A^*$ to closed densely defined linear operators on ${\mcH}$.  Furthermore, the injective nature of ${\bf J}$ and that of $\mcB \to \mcB'$ means that $A^*A$ extends to an accretive self adjoint linear operator on  ${\mcH}$. 

If  ${A} \in L[\mcB]$, then $\left( {A^* Au,u} \right)  \geqslant 0$, and  
\[
\begin{gathered}
  \left( {A^* Ag\,,\,\,u} \right)  = \left\langle {A^* Au\,,\,\,{\mathbf{J}} (u)} \right\rangle  = \left\langle {{\mathbf{J}}^{ - 1} A'{\mathbf{J}} (Au)\,,\,\,{\mathbf{J}} (u)} \right\rangle  \hfill \\
   = \left\langle {{\mathbf{J}} (Au)\,,\,\,Au} \right\rangle  = \left\langle {u,\,\,A'{\mathbf{J}} (Au)} \right\rangle  = \left\langle {{\mathbf{J}} (u)\,,\,\,A^* Au} \right\rangle  = \left( {u\,,\,\,A^* Au} \right) \hfill \\ 
\end{gathered} 
\]
for all $u \in \mcB$.  By Lax's Theorem, $A^*A$ has a bounded extension to $\mcH$ and $\left\| {A^* A} \right\|  = \left\| A \right\|^2  \leqslant k\left\| A \right\|_B^2 $, where $k$ is a positive constant.
\end{proof}
The above theorem shows that every closed densely defined linear operator on $\mcB$ has a closed densely defined extension to $\mcH$, with bounded operators on $\mcB$ becoming bounded on $\mcH$.  The proof of the next result follows.
\begin{cor}If $\mcB$ be a separable reflexive Banach space and $ A \in \mcC[\mcB]$, then 
\begin{enumerate}
\item the operator $
A^ * A \ge 0$  (accretive),		
\item $(A^ * A)^ *   = A^ * A$ (naturally selfadjoint), and 
\item $I + A^ * A$ has a bounded inverse.
\end{enumerate}
\end{cor} 
\begin{thm} If $\mcB$ be a separable reflexive Banach space, then every $ A \in \mcC[\mcB]$ is of Baire class one.
\end{thm}
\begin{proof}Since each $A \in \mcC[\mcB]$ has a extension $\bar{A}$ to $\mcC[\mcH]$, we can write $A=UT$, where $UT$ is the restriction to $\mcB$ of the polar decomposition of $\bar{A}$.  From $A^*=TU^*$, we see that $T=A^*U$, so that 
\[
AA^* =(UT)(TU^*)= UA^*AU^* \, \Rightarrow \, AA^*U=UA^*A.
\]
It follows that, $A=\bar{T}U$ and $A^*=U^*\bar{T}$.

For $\la>0$, let $R(\la, -T)$ be the resolvent of $-T$ and let 
\beqn
A_\la= \la AR(\la, -T)=\la^2 UR(\la,-T)-\la U.
\eeqn
It is easy to see that $A_\la$ is bounded and that $AR(\la, -T)\phi=R(\la, -\bar{T})A\phi$ for $\phi \in D(A)$.

If $0<\la$ and $\phi \in D(A)$ we have
\[
\begin{gathered}
  R(\lambda ,-T)\left( {\lambda I +T} \right)\phi = \phi\; \Rightarrow  \hfill \\
  \lambda R(\lambda ,-T)\phi - \phi = R(\lambda ,-T)(-T\phi) \; \Rightarrow  \hfill \\
  \left\| {\lambda R(\lambda ,-T)\phi - \phi} \right\| \leqslant \left\| {R(\lambda ,-T)} \right\|\left\| {T\phi} \right\| \leqslant \lambda ^{ - 1} \left\| {T\phi} \right\| \hfill \\ 
\end{gathered} 
\]
This last term converges to zero as $\la \to \iy$, so that 
\[
\mathop {\lim }\limits_{\lambda  \to \infty } \lambda R(\lambda ,-T)\phi= \phi.
\]
Since $D(A)$ is dense, the convergence holds for all $\phi \in \mcB$.  

For the second part, we see from the last result that
\[
\mathop {\lim }\limits_{\lambda  \to \infty } \lambda AR(\lambda ,-T)\phi = \mathop {\lim }\limits_{\lambda  \to \infty } \lambda R(\lambda ,-\bar{T})A\phi = A\phi,
\]
whenever $\phi \in D(A)$.
\end{proof}
\subsection*{Examples} 
We begin with the following useful result (see Kato \cite{K}, pg. 168).
\begin{thm}Let $T$ be a densely defined linear operator on a reflexive Banach space $\mcB$.  Then, the following holds:
\begin{enumerate}
\item The adjoint of $T, \; T^*$ is a closed linear operator.
\item The operator $T$ has a closed extension if and only if $D(T^*)$ is dense in $\mcB^*$.  In this case, the closure $\bar{T}= T^{**}$.
\item If $T$ is closable, then $(\bar{T})^*=T^*$.
\end{enumerate}
\end{thm}
From the last result, we see that, for any closed densely defined linear operator $A$ defined on $L^p[\R^n], \; 1 < p< \iy$, for which the domain of $A^*,\; D(A^*) \subset L^q[\R^n], 1/p +1/q=1$, is dense in $L^p[\R^n]$, also has a  closed densely defined extension to $L^p[\R^n]$. 
\begin{ex}Let $A$ be a second order differential operator on $L^p[\R^n]$, of the form
\[
A = \sum\limits_{i,j = 1}^n {{a_{ij}}({\mathbf{x}})} \frac{{{\partial ^2}}}{{\partial {x_i}\partial {x_j}}} + \sum\limits_{i,j = 1}^n {{b_{ij}}({\mathbf{x}})} {x_j}\frac{\partial }{{\partial {x_i}}},
\]
where ${\bf a}({\bf x})=\left[\kern-0.15em\left[ {{a_{ij}}({\mathbf{x}})} 
 \right]\kern-0.15em\right]$ and ${\bf b}({\bf x})=\left[\kern-0.15em\left[ {{b_{ij}}({\mathbf{x}})} \right]\kern-0.15em\right]$ are matrix-valued functions in $\C_c^{\iy}[\R^n \times \R^n]$ (infinitely differentiable functions with compact support).  We also assume that, for all ${\bf x} \in \R^n \; det\left[\kern-0.15em\left[ {{a_{ij}}({\mathbf{x}})} 
 \right]\kern-0.15em\right] > \e$  and the imaginary part of the eigenvalues of ${\bf b}({\bf x})$ are bounded above by $-\e$, for some $\e>0$.
Note, since we don't require  ${\bf a}$ or ${\bf b}$ to be symmetric, $A \ne A^*$.

It is well-known that $\C_c^{\iy}[\R^n]$ is dense in  $L^p[\R^n] \cap L^q[\R^n]$ for all $p,q \in [1, \iy) \cap \N$.  Furthermore, since $A^*$ is invariant on $\C_c^{\iy}[\R^n], \; A^*: L^p[\R^n] \to  L^p[\R^n]$. It now follows from Theorem 2.8, that  $A^*$ has a closed densely defined extension to $L^p[\R^n]$.
\end{ex} 
\begin{ex}
The second example shows directly the setup used in Theorem 2.3. Let $\Om$ be a bounded open domain of class $\C^1$ in $\R^n$ and let $\mcH_0^1[\Om]$, the set of all real-valued functions $u \in L^2[\Om]$ such that their first order weak partial derivatives are in $L^2[\Om]$ and vanish on the boundary.  It follows that 
\[\left( {u,v} \right) = \int_\Omega  {\nabla u({\mathbf{x}}) \cdot \nabla v({\mathbf{x}})d{\mathbf{x}}} \]
defines an inner product on $\mcH_0^1[\Om]$.  The  dual ${\mcH^{ - 1}}[\Omega ]$ coincides with the set of all distributions of the form
\[
u = {h_0} + \sum\limits_{i = 1}^n {\frac{{\partial {h_i}}}{{\partial {x_i}}}} ,\quad {\text{where}}\;{h_i} \in {L^2}[\Omega ],\quad 1 \leqslant i \leqslant n.
\] 
In this case we also have for $p \in [2, \iy)$ and $q \in (1, 2], \tf{1}{p}+\tf{1}{q}=1$ that,  
\[
\mcH_0^1[\Omega ] \subset {L^p}[\Omega ] \subset {L^q}[\Omega ] \subset {\mcH^{ - 1}}[\Omega ]
\]  
all as continuous dense embeddings.

From the inner product on $\mcH_0^1[\Om]$ we see that $J_0=-\De$, the Laplace operator under Dirichlet homogeneous boundary conditions on $\Om$.  If we set $\mcH=\mcH^{-1}$ and $J=J_0^{-1}$, we can apply Theorem 2.4 to obtain $A^* \in \mcC[L^r(\Om)]$, for all $1<r<\iy$, by $A^* =-\De A'[-\De]^{-1}$, for all $A' \in \mcC[L^{r' }(\Om)]$.  Where $\tf{1}{r}+ \tf{1}{r'}=1$.
\end{ex}
\subsection{Spectral Theorem}
\begin{thm} Let $A \in {\mcC}[{\mcB}]$ be arbitrary, where $\mcB$ is a separable reflexive Banach space.  Then, for each $\phi \in D(A)$, there exists a  deformed spectral measure ${\bf{F}}(\cdot)$ and a  vector-valued function ${\bf{F}}(\lambda )\phi$ of bounded variation such that:
\begin{enumerate}
\item $D(A)$  satisfies
\[
D(A) = \left\{ {\left. {\phi \in \mathcal{B}} \ \right|\;\int_{0}^\iy {\lambda ^2 \left({d{\bf{F}}(\lambda )\phi,S_\phi} \right)_\mathcal{B}  < \infty } } \right\}
\]
and
\item 
\[
A\phi =  \int_{ 0}^\iy {\lambda d{\bf{F}}(\lambda )\phi}, \ {\rm for \ all} \; \phi \in D(A). 
\]
\item If $g(\cdot)$ is a complex-valued Borel function defined (a.e) on $\R$, then 
\[
g(A)\phi= \int_{ 0}^\iy {g(\lambda) d{\bf{F}}(\lambda )\phi}  \ {\rm for \ all} \; \phi \in D_g(A),
\]
where
\[
D_g(A) = \left\{ {\left. {\phi \in \mathcal{B}} \ \right|\;\int_{0}^\iy {\lt|g(\lambda)\rt|^2 \left( {d{\bf{F}}(\lambda )\phi,S_\phi} \right)_\mathcal{B}  < \infty } } \right\}.
\]
\end{enumerate}
\end{thm}
\subsection*{Conclusion} 
A major result of this paper is the discovery of a new direct spectral type representation for the family of closed densely defined linear operators on a Hilbert space, which we call the deformed representation.  The advantage of this approach is that it has a similar extension to the family of closed densely defined linear operators on a separable reflexive Banach space. 
\acknowledgements
During the course of the development of this work, we have benefited from important critical remarks from Professor Ioan I. Vrabie.  We would like to sincerely thank Professor Anatolij Pliczko for important correspondence on  operators of Baire class on Banach spaces.  

\end{document}